\documentclass[pra,twocolumn,amsmath,amssymb,nofootinbib,superscriptaddress]{revtex4-2}
\usepackage[utf8]{inputenc}
\usepackage[dvipsnames]{xcolor}
\usepackage{amsfonts}
\usepackage{amsmath}
\usepackage{amssymb}
\usepackage{amsthm}
\usepackage{svg}
\usepackage{qcircuit}
\usepackage{xfrac} 
\usepackage{graphicx}

\usepackage{bm}
\usepackage[makeroom]{cancel}
\usepackage[colorlinks=true,citecolor=cyan]{hyperref}
\usepackage[capitalise]{cleveref}
\crefname{section}{Sec.}{Secs.}
\crefname{table}{Tab.}{Tabs.}
\crefname{figure}{Fig.}{Figs.}
\crefname{definition}{Def.}{Defs.}
\crefname{lema}{Lem.}{Lems.}
\crefname{theorem}{Thm.}{Thms.}
\crefname{corollary}{Cor.}{Cors.}

\usepackage{comment}

\newtheorem{theorem}{Theorem}
\newtheorem{definition}{Definition}
\newtheorem{lemma}{Lemma}
\newtheorem{corollary}{Corollary}

\newcommand{\galg}{\mathfrak{g}}
\newcommand{\kalg}{\mathfrak{k}}
\newcommand{\malg}{\mathfrak{m}}
\newcommand{\halg}{\mathfrak{h}}
\newcommand{\ham}{\mathcal{H}}

\begin{document}

\title{Algebraic Compression of Quantum Circuits for Hamiltonian Evolution}

\author{Efekan K\"okc\"u}
\affiliation{Department of Physics, North Carolina State University, Raleigh, North Carolina 27695, USA}

\author{Daan Camps}
\affiliation{Computational Research Division,
            Lawrence Berkeley National Laboratory,
            Berkeley, CA 94720, USA}
            
\author{Lindsay Bassman}
\affiliation{Computational Research Division,
            Lawrence Berkeley National Laboratory,
            Berkeley, CA 94720, USA}
            
\author{J.~K.~Freericks}
\affiliation{Department of Physics, Georgetown University, 37th and O Sts. NW, Washington, DC 20057 USA}

\author{Wibe~A.~de~Jong}
\affiliation{Computational Research Division,
            Lawrence Berkeley National Laboratory,
            Berkeley, CA 94720, USA}
            
\author{Roel~Van~Beeumen}
\affiliation{Computational Research Division,
            Lawrence Berkeley National Laboratory,
            Berkeley, CA 94720, USA}
            
\author{Alexander F. Kemper}
\affiliation{Department of Physics, North Carolina State University, Raleigh, North Carolina 27695, USA}

\date{\today}

\begin{abstract}
Unitary evolution under a time dependent Hamiltonian is a key component
of simulation on quantum hardware. Synthesizing the corresponding
quantum circuit is typically done by breaking the evolution into
small time steps, also known as Trotterization, which leads to
circuits whose depth scales with the number of steps. 
When the circuit elements are limited to
a subset of $SU(4)$ --- or equivalently, when the Hamiltonian
may be mapped onto free fermionic models ---
several identities exist that combine and simplify
the circuit. Based on this, we
present an algorithm that compresses the Trotter steps into
a single block of quantum gates.
This results in a fixed depth
time evolution for certain classes of Hamiltonians. 
We explicitly show how this algorithm works for
several spin models, and demonstrate
its use for adiabatic state preparation of the
transverse field
Ising model.
\end{abstract}

\maketitle

\section{Introduction}

Quantum computers were initially conceived of as a tool for simulating quantum systems \cite{feynman1982simulating}, and indeed, this is seen as one of the most promising applications for near-term quantum processors \cite{preskill2018quantum}. The main challenge for implementing such simulations on a quantum computer is to generate a circuit---usually comprised of single- and two-qubit gates---that mimics the effect of the time evolution operator $U(t)$ on qubits, which is termed \emph{unitary synthesis}. Unitary synthesis in the most generic case is exponentially hard due to exponential growth of the Hilbert space dimension, i.e. size of the matrix $U(t)$, with the system size \cite{khaneja2005constructive,drury2008quantum_shannon}. Though there are ways to ameliorate this problem by leveraging the algebra generated by the terms in the Hamiltonian \cite{kokcu2021fixed}, generating an exact circuit remains classically NP-hard for interacting fermion models.   
In practice, an approximate circuit implementation of the unitary is often sufficient
and can even be imposed by technological constraints as we discuss next.

The challenges are exacerbated by the fact that near-term quantum computers are small and noisy \cite{preskill2018quantum}, which greatly limits how deep quantum circuits can be before their results are indistinguishable from random noise. This means that circuit synthesis techniques must also attempt to optimize quantum circuits to be as shallow as possible;
such circuit optimization is an NP-hard problem for general quantum circuits \cite{herr17, botea2018}.

Generating a circuit is particularly difficult for dynamic simulations which simulate a quantum system evolving in time.  This is because a separate circuit must be generated for each time-step, and for a general Hamiltonian, the circuit depth increases with increasing simulation time \cite{berry2007efficient, childs2009limitations,Childs2021}.  

Here, we present an algorithm for compressing circuits for simulations of time-evolution to a constant depth, independent of increasing simulation time, for a particular class of Hamiltonians. The common feature of these Hamiltonians is that they can be mapped to free fermionic Hamiltonians \cite{chapman2020characterization}; equivalently, they produce a Hamiltonian algebra that scales polynomially in the number of qubits \cite{kokcu2021fixed}. This property
allows us to circumvent the no-fast-forwarding theorem \cite{berry2007efficient,atia2017fast,gu2021fast}.

Our approach uses a first order Trotter-Suzuki decomposition of the time-evolution operator
of a time-dependent Hamiltonian $\ham(t)$ for the unitary synthesis step.
The unitary operator $U(t)$ for this process obeys 
\begin{align}
    \partial_t U(t) = -i \ham(t) U(t)
\end{align}
with $U(t=0) = \mathcal{I}$, i.e. the identity operator. This equation can be approximately solved by Trotter decomposition with
a time step size $\Delta t$:
\begin{align}\label{eq:solution}
    U(t) = \lim_{\Delta t \rightarrow 0} e^{-i\Delta t \ham(t_N)} \cdots e^{-i\Delta t \ham(t_2)}e^{-i\Delta t \ham(t_1)},
\end{align}
where $t_n = (n-1)\Delta t$, and $t_N = t$ (kept constant while taking the limit). After approximating each small time step $e^{-i\Delta t \ham(t_i)}$, the right hand side of Eq.~\eqref{eq:solution} often has a straightforward circuit implementation that approximates $U(t)$.
The accuracy of the Trotter decomposition depends on the size of $\Delta t$. This is the trade-off with the Trotter decomposition:
the smaller the time step, the longer the circuit. 
Moreover, regardless of the choice of $\Delta t$, the circuit
grows with increasing simulation time.

In order to deal with the circuits that grow as a function of
simulation time, several approaches exist, each with inherent
limitations. We limit the discussion
to those methods that can be executed on Noisy Intermediate Scale Quantum (NISQ) hardware, putting
aside the others \cite{Berry2015,Low2017, Low2018, Kalev2020, Haah2021, cirstoiu2020variational}.
General purpose transpilers
can be used to shorten the growing-depth circuits, but these do
not typically take advantage of the algebraic structure of the
problem, and are limited to how much compression can be achieved,
ultimately placing a limit on the largest feasible simulation time.
Constant depth circuits can be obtained variationally by
optimizing a generic circuit via a hybrid classical-quantum
algorithm \cite{commeau2020variational}. However, these approximate circuits are only applicable to time evolution of
a particular initial state and their error grows with simulation time. Alternatively, a more general constant depth circuit may be found by using the distance from a specific target unitary as the cost
function \cite{bassman2021constant} --- this approach does not scale favorably in the number of qubits due to the necessarily large
number of parameters in the classical optimization. Finally, constructive approaches
to circuit synthesis that rely on the algebraic structure of the
unitaries have been proposed \cite{khaneja2005constructive,earp2005constructive,kokcu2021fixed}, but these require a similar classical
optimization step. 

Our proposed method, while not applicable to all problems, is a constructive proof of the circuit ansatz used in Ref.~\cite{bassman2021constant}. Because of the constructive nature, it scales to 1,000s of qubits, which
is out of the reach of previous methods
 \cite{bassman2021constant,kokcu2021fixed} that require optimization.
This work is accompanied by a dual paper \cite{simax} which focuses
on the underlying matrix structures of the circuit elements, and on the efficient and numerically stable implementation of the circuit compression algorithms for each model discussed below. We show that this compression technique scales polynomially in the number of qubits, and is thus applicable for the foreseeable future on quantum hardware.
Our codes are made available as part of the fast free fermion compiler (\texttt{F3C}) \cite{f3c, f3cpp} on \url{https://github.com/QuantumComputingLab}.
\texttt{F3C} is build based on the \texttt{QCLAB} toolbox \cite{qclab,qclabpp} for creating and
representing quantum circuits.

We demonstrate the power of our compressed circuits by showing how they enable adiabatic state preparation (ASP), where a system is evolved under a slowly varying Hamiltonian in order to prepare a non-trivial ground state.  The system is initialized in the ground state of the initial Hamiltonian, which is presumed to be trivial to prepare, and ends in the (non-trivial) ground state of the final Hamiltonian given slow enough variation of the Hamiltonian. 
Given the constant depth nature of the circuit for different times, the state preparation for long times, is ultimately accomplished in a single step.
We apply this to preparing
the ground state of the transverse field Ising model (TFIM) on  $5$ qubits.

The paper is organized as follows. We define an object called a ``block'' and prove the existence of a constructive compression method for circuits made only of blocks in 
\cref{sec:compression_theorem}. In \cref{sec:model_specific},
we propose specific mappings between the circuit elements
of Trotter-decomposed time evolution operators and blocks for the Kitaev chain, the TFIM, the transverse field XY (TFXY) model, as well as a nearest neighbor free fermion model.
We further prove that the three properties of blocks hold for each,
thus showing that a compressed circuit exists for these models.
We apply our method to ground state construction of the TFIM with open boundary conditions in 
\Cref{sec:specific_example}, where we show that compressed Trotter circuits are superior to traditional Trotter circuits due to their ability to simulate as slow time evolution as desired, which is significant for adiabatic evolution.
We finish by discussing the potential implications of
our results in a broad sense in \cref{sec:discussion}.

\section{Blocks and the Compression Theorem}
\label{sec:compression_theorem}

\begin{figure*}[t]
    \centering
    \includegraphics[width = 1.35\columnwidth]{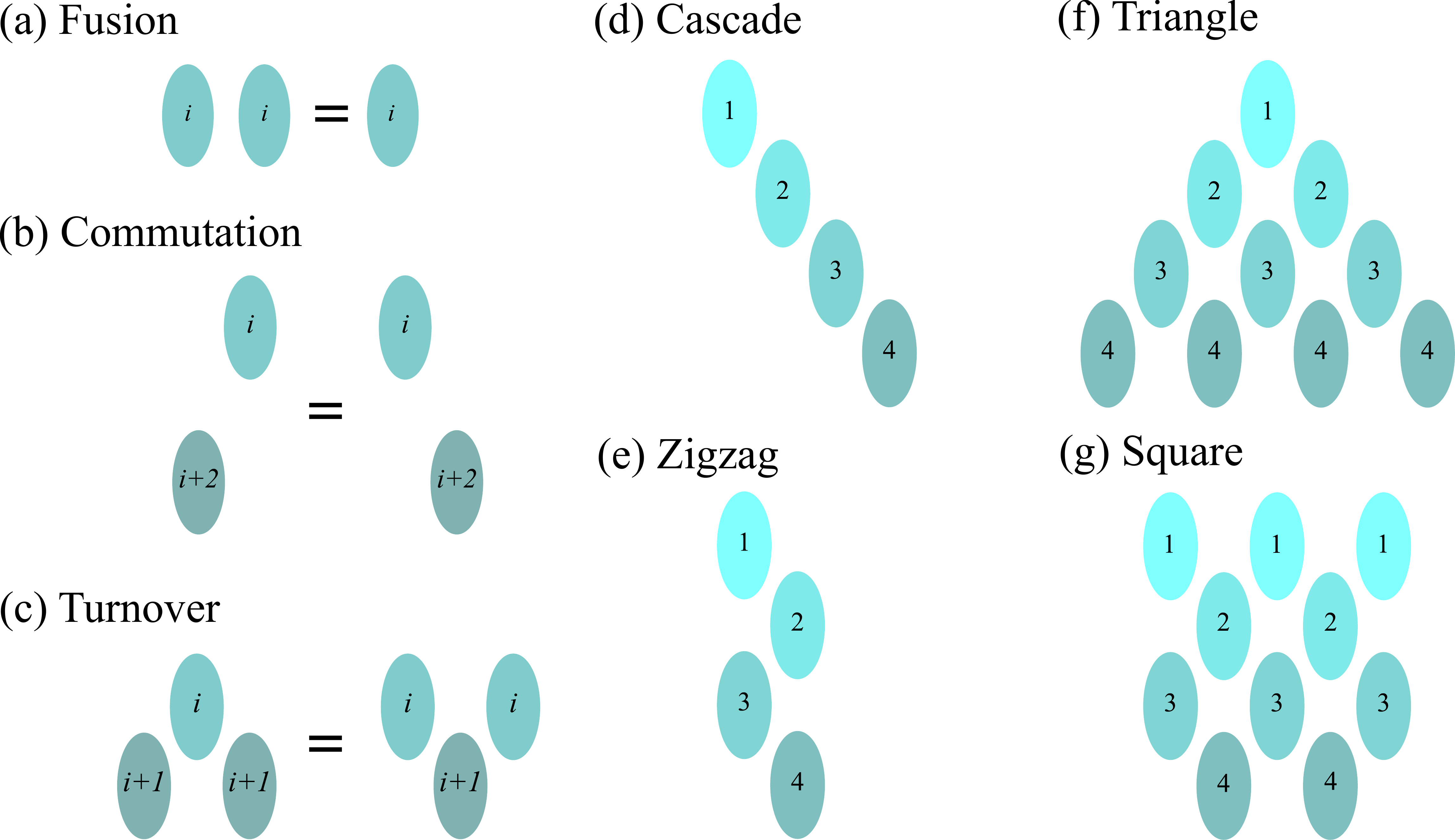}
    \label{fig:blocks}
    \caption{Panels \textbf{(a)}, \textbf{(b)}, \textbf{(c)} illustrate the defining properties of a block given in \cref{block}. Panels \textbf{(d)}, \textbf{(e)}, \textbf{(f)}, \textbf{(g)} } illustrate the block structures cascade, zigzag and triangle defined in \cref{block,cascade,triangle}, and square defined in \cref{thm:triangle_to_square}.
\end{figure*}

In this section, we will define and prove generic mathematical tools necessary to show that certain models have a Trotter expansion that can be compressed down to a fixed depth circuit. To compress the Trotter expansion, we map certain elements in the Trotter step into a mathematical structure called a ``block''. These block structures are model dependent and are required to have certain simplification rules given in \cref{block}. We follow with lemmas and theorems that use these simplification rules to compress a Trotter circuit that is fully expressed via blocks into a fixed depth circuit.

Let $B_i(\vec{\theta})$ be functions of $\vec{\theta}$, where the index can
be any positive integer $i=1,2,3,\ldots$ and $\vec{\theta}$ are multiple real
numbers.
Each $B_i(\vec{\theta})$ will be mapped to quantum circuit elements in a model dependent way such that the parameters $\vec{\theta}$ will be the parameters for the gates, and the index will separate them in terms of which qubits they act on. The number of parameters and index to qubit(s) mapping vary from model to model, as it will be seen in \cref{sec:model_specific}. For example, we will have $B_i(\theta) \equiv \exp(-i\theta X_i X_{i+1})$ for odd $i$ and $B_i(\theta) \equiv \exp(-i\theta Y_i Y_{i+1})$ for even $i$ as the mapping for an open one-dimensional (1D) Kitaev chain in \cref{1D_Kitaev_Chain}, and show that
this mapping satisfies the three defining properties listed below.

\begin{definition}[Block]
\label{block}
Define a ``block'' $B_i(\vec{\theta})$ as a structure that satisfies the following three properties:
\begin{enumerate}
  \item \textbf{the fusion property:} for any set of parameters $\vec{\alpha}$ and $\vec{\beta}$, there exist $\vec{a}$ such that
        \begin{equation}
          B_i(\vec{\alpha}) \, B_i(\vec{\beta}) = B_i(\vec{a}),
        \end{equation}
  \item \textbf{the commutation property:} for any set of parameters $\vec{\alpha}$ and $\vec{\beta}$
        \begin{equation}
          B_i(\vec{\alpha}) \, B_j(\vec{\beta}) = B_j(\vec{\beta}) \, B_i(\vec{\alpha}), \qquad |i-j|>1,
        \end{equation}
  \item \textbf{the turnover property:} for any set of parameters $\vec{\alpha}$, $\vec{\beta}$ and $\vec{\gamma}$ there exist $\vec{a}$, $\vec{b}$ and $\vec{c}$ such that 
        \begin{equation}
          B_i(\vec{\alpha}) \, B_{i+1}(\vec{\beta}) \, B_i(\vec{\gamma}) = B_{i+1}(\vec{a}) \, B_i(\vec{b}) \, B_{i+1}(\vec{c}).
        \end{equation}
\end{enumerate}
\end{definition}

Henceforth, we drop the parameters $\vec{\theta}$ when referring to a block and simply write $B_i$ for convenience. The reader should assume that each block in an expression may have different parameters unless it is stated otherwise. 

Using \cref{block}, we now define three important ``block'' structures, which will be used to obtain and prove the fixed depth circuits.
We start with a cascade of blocks.

\begin{definition}[Cascade]
\label{cascade}
Define a ``cascade'' of blocks to be
\begin{equation}
  C_{i,j} := \prod_{k=i}^j B_k = B_i \, B_{i+1} \,\cdots\, B_j,
  \label{eq:cascade}
\end{equation}
where $B_i$ are blocks as defined in \cref{block} and $i<j$.
\end{definition}

\begin{lemma}
\label{lem:cascblock}
Let $B_m$ be a block as defined in \cref{block} and $C_{j,n}$ a cascade of blocks
as defined in \cref{cascade} with $j \leq m < n$.
Then the transformation
\begin{equation}
C_{j,n} \, B_m = B_{m+1} \, C_{j,n},
\end{equation}
can be performed by one turnover operation. 
\end{lemma}
\begin{proof}
Using \cref{cascade}, we obtain
\begin{align*}
  C_{j,n} B_m &= \left( \prod_{k=j}^n B_k \right) B_m, \\
              &= \left( \prod_{k=j}^{m+1} B_k \right) B_m
                 \left( \prod_{k=m+2}^{n} B_k \right), \\
              &= \left( \prod_{k=j}^{m-1} B_k \right) B_m \, B_{m+1} \, B_m
                 \left( \prod_{k=m+2}^{n} B_k \right).
\end{align*}
Next, applying the turnover property to $B_m \, B_{m+1} \, B_m$ and using the commutation property, yields
\begin{align*}
  C_{j,n} B_m &= \left( \prod_{k=j}^{m-1} B_k \right) B_{m+1} \, B_m \, B_{m+1}
                 \left( \prod_{k=m+2}^{n} B_k \right), \\
              &= B_{m+1} \left( \prod_{k=j}^{m-1} B_k \right)
                         \left( \prod_{k=m}^{n} B_k \right), \\
              &= B_{m+1} \left( \prod_{k=j}^{n} B_k \right), \\
              &= B_{m+1} \, C_{j,n},
\end{align*}
which completes the proof.
\end{proof}

Next, we define a triangle of blocks, comprised of multiple cascades,
and prove that an extra block can be fully compressed into a triangle.

\begin{definition}[Triangle]
\label{triangle}
Define a ``triangle'' of blocks as
\begin{equation}
  T_i := \prod_{k=0}^{i-1} C_{i-k,i} = C_{i,i} \, C_{i-1,i} \,\cdots\, C_{1,i},
  \label{eq:triangle}
\end{equation}
where $C_{i,j}$ are cascades of blocks as defined in \cref{cascade}.
\end{definition}

\begin{figure*}[t]
    \centering
    \includegraphics[width = 2\columnwidth]{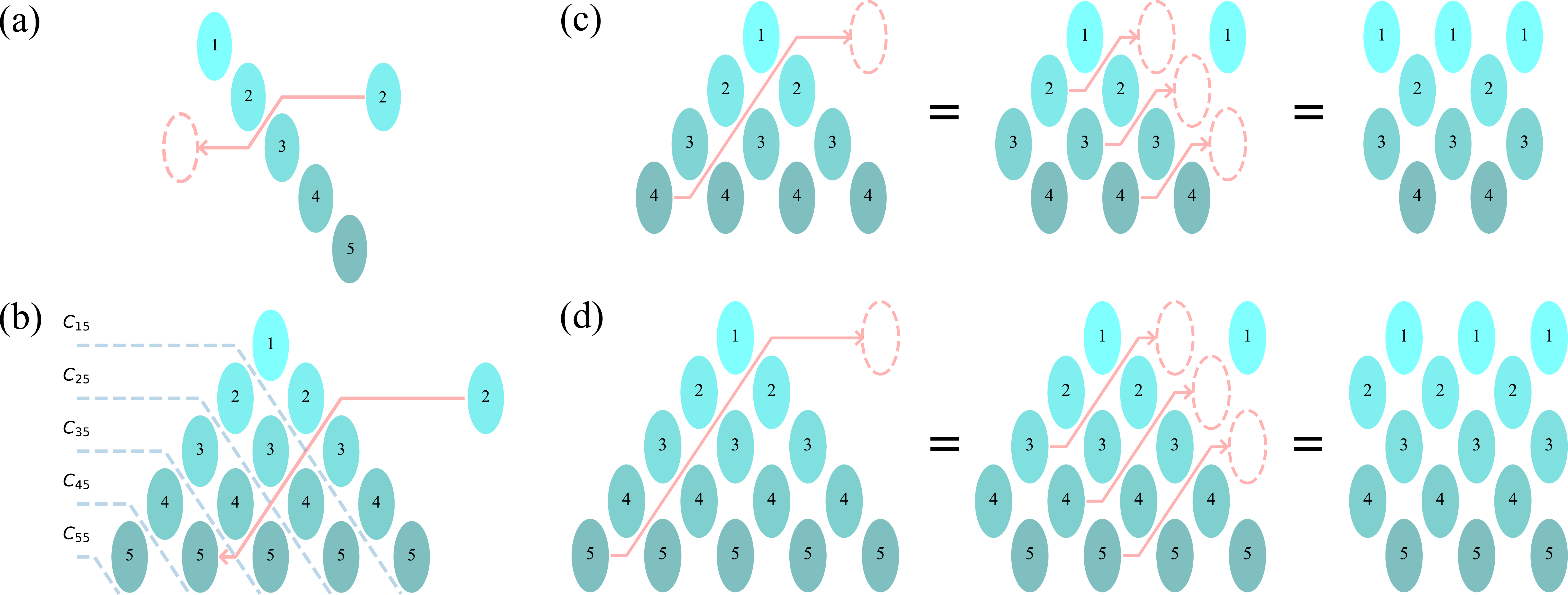}
    \label{fig:triangletosquare}
    \caption{Panel \textbf{(a)} illustrates \cref{lem:cascblock}. As it can be seen here, this lemma can be interpreted as a block shifting one position down while passing through a cascade from right to left. Panel \textbf{(b)} includes the diagrammatic representation of \cref{thm:blocks}. Block $B_2$ on the right side of a triangle with height $5$ i.e. $T_5$ finds its way all the way down to $5$th level passing through cascades and finally merges with the block that is at the end of the path, and results with a triangle with height $5$ made by blocks with different parameters at the end. Panels \textbf{(c)} and \textbf{(d)} illustrate the transformation described in \cref{thm:triangle_to_square} respectively for triangles with even and odd height. As described, one starts by pushing the leftmost cascade all the way to the right, and continue by pushing every other cascade until reaching to the longest cascade $C_{1,n}$. This results in the square circuits shown in the rightmost part of both panels.}
\end{figure*}

\begin{theorem}
\label{thm:blocks}
Let $B_m$ be a block as defined in \cref{block} and $T_n$ a triangle of blocks
as defined in \cref{triangle} with $m \leq n$.
Then $B_m$ can be fully merged into $T_n$, i.e.,
\begin{equation}
T_n \, B_m = T_n,
\end{equation}
via $n-m$ turnover operations and $1$ fusion operation.
\end{theorem}
\begin{proof}
Using \cref{triangle}, we obtain
\begin{equation*}
  T_n \, B_m = \left( \prod_{k=0}^{n-1} C_{n-k,n} \right) B_m
             = C_{n,n} \, C_{n-1,n} \,\cdots\, C_{1,n} \, B_m.
\end{equation*}
Applying \cref{lem:cascblock} $n-m$ times, yields
\begin{align*}
  T_n \, B_m &= C_{n,n} \, C_{n-1,n} \,\cdots\, C_{1,n} \, B_m, \\
             &= C_{n,n} \, C_{n-1,n} \,\cdots\, C_{2,n} \, B_{m+1} \, C_{1,n},\\
             &\ \, \vdots \\
             &= C_{n,n} \, C_{n-1,n} \,\cdots\, C_{n-m+1,n} \, B_n \, C_{n-m,n}
                \,\cdots\, C_{1,n},
\end{align*}
where each time we swap a block with a cascade, we perform $1$ turnover.
Hence, we have done $n-m$ turnover operations in order to move the extra
block $B_m$ to $B_n$.
Next, we can merge $B_n$ into the triangle
$C_{n,n} \, C_{n-1,n} \,\cdots\, C_{n-m+1,n}$, i.e.,
\begin{align*}
  C_{n-m+1,n} \, B_n &= \left( \prod_{k=n-m+1}^{n} B_k \right) B_n, \\
                     &= \left( \prod_{k=n-m+1}^{n-1} B_k \right) B_n \, B_n,\\
                     &= \left( \prod_{k=n-m+1}^{n-1} B_k \right) B_n
                      = C_{n-m+1,n},
\end{align*}
where we applied the fusion property once. Finally,
\begin{align*}
  T_n \, B_m &= C_{n,n} \, C_{n-1,n} \,\cdots\, C_{n-m+1,n} \, C_{n-m,n}
                \,\cdots\, C_{1,n}, \\
             &= \prod_{k=0}^{n-1} C_{n-k,n}
              = T_n,
\end{align*}
via $n-m$ turnover and 1 fusion operation.
\end{proof}

\Cref{thm:blocks} shows that a triangle and any finite set of blocks can be merged into a triangle. Here we define another block structure, called zigzag of blocks, which corresponds to a
time step of a Trotter decomposition of the Hamiltonian evolution after mapping to blocks. We use \cref{thm:blocks} to explicitly show that a zigzag can efficiently be compressed into a triangle of
appropriate size. We also show that a triangle can be merged with another triangle efficiently, using the fact that a triangle is a finite set of blocks.

\begin{definition}[Zigzag]
\label{zigzag}
Define a ``zigzag'' of blocks as
\begin{equation}
  L_{i,j} := \Bigg( \prod_{\substack{k=i \\ \mathrm{odd}}}^j  B_k \Bigg)
            \Bigg( \prod_{\substack{k=i \\ \mathrm{even}}}^j B_k \Bigg),
  \label{eq:zigzag}
\end{equation}
where $B_k$ are blocks as defined in \cref{block} and $i<j$; note all even numbered blocks are to the right of odd numbered blocks.
\end{definition}

\begin{corollary}
\label{trianglezigzag}
Let $C_{1,n}$, $T_n$, and $L_{1,n}$ be a cascade, triangle, and zigzag of blocks as defined in \cref{cascade,triangle,zigzag}, respectively.
Then a triangle with height $n$ and a cascade (zigzag) with height $n$ can be merged into a triangle with height $n$, i.e.
\begin{align}
  T_n \, C_{1,n} &= T_n, &
  T_n \, L_{1,n} &= T_n,
\end{align}
with $n(n-1)/2$ turnover and $n$ fusion operations.
\end{corollary}
\begin{proof}
In both cases, we are merging one block from every height $1,2,3,..,n$, therefore we use the result of \cref{thm:blocks} for each of these $n$ blocks. Each block will require $1$ fusion, therefore we need to apply $n$ fusion operations. The number of turnover operations is
\begin{align*}
    \sum_{k=1}^n (n-k) = \sum_{k=1}^{n-1} k = \frac{n(n-1)}{2},
\end{align*}
which concludes the proof.
\end{proof}

\begin{corollary}
\label{triangletriangle}
Two triangles with height $n$ can be merged into a triangle with height $n$, i.e.
\begin{equation}
  T_n \, T_n = T_n,
\end{equation}
via $n(n^2-1)/6 = O(n^3)$ turnover and $n(n-1)/2 = O(n^2)$ fusion operations.
\end{corollary}
\begin{proof}
There are $n(n-1)/2$ blocks in $T_n$. As a result of \cref{thm:blocks} this merging will require $n(n-1)/2$ fusion operations. Using the definition of $T_n$, we see that we are effectively merging the triangle with cascades with heights $1$ to $n$. Using the result of \cref{trianglezigzag}, we get the number of turnover operations required for this merger to be
\begin{align}
    \sum_{k=1}^n \frac{k(k-1)}{2} = \frac{n(n^2-1)}{6},
\end{align}
which concludes the proof.
\end{proof}

Using \cref{trianglezigzag,triangletriangle} as induction steps, we provide the two following compression techniques for Trotter decompositions involving blocks only:

\begin{corollary}[Compression for time-dependent Hamiltonians]
\label{timedependentcompression}
    A product of $r$ zigzags with height $n$ for $r>1$ will lead to a triangle with height $n$
    \begin{align}
        \prod_{k=1}^r L_{1,n} = T_n,
    \end{align}
    via $O(n^2 r)$ turnover and $O(n r)$ fusion operations.
\end{corollary}

\begin{proof}
A zigzag can be viewed as a triangle with identity operations in place of the missing blocks. Therefore for $r=1$, the result should be a triangle. Using \cref{trianglezigzag}, we know that adding another zigzag to the triangle each time will require $O(n^2)$ turnover and $O(n)$ fusion operations. Doing this $r$ times will require $O(n^2 r)$ turnover and $O(n r)$ fusion operations.
\end{proof}

\begin{corollary}[Compression for time-independent Hamiltonians]
\label{timeindependentcompression}
    Consider $r$ identical zigzags with height $n$. Using the fact that they are same, they can be compressed into a triangle with height $n$ 
    \begin{align}\label{eq:timeindependentcompression}
        (L_{1,n})^r = T_n
    \end{align}
    where we use power instead of the product symbol to emphasize that the parameters are the same, via $O(n^3 \log_2 r)$ turnover and $O(n^2 \log_2 r)$ fusion operations.
\end{corollary}

\begin{proof}
Given that zigzags (with added identity gates) are triangles , the left hand side of \eqref{eq:timeindependentcompression} can also be written as $(T_n)^r$. We can replace it with $((T_n)^2)^{r/2}$ and turn it into
$(T_n)^{r/2}$ via $O(n^2)$ fusion and $O(n^3)$ turnover operations using \cref{triangletriangle}. We need to repeat this  $O(\log_2 r)$ many steps in total, leading to this compression being performed via $O(n^3 \log_2 r)$ turnover and $O(n^2 \log_2 r)$ fusion operations.
\end{proof}

\Cref{timedependentcompression,timeindependentcompression} are the main results used for compression. For a time-independent Hamiltonian, one can choose either of the two compression algorithms.
The performance depends on the number of Trotter steps needed, and the number of qubits.
Depending on the particular problem at hand,
one or the other may be preferable;
the $O(n^2 r)$ scaling of \cref{timedependentcompression} is a
better choice for larger systems, whereas \cref{timeindependentcompression} is a better choice for a large number of Trotter steps.
For time-dependent Hamiltonians, each Trotter time step will have different coefficients, meaning the corresponding zigzags will not all be identical.  In this case, \cref{timedependentcompression} is the only choice for compression.

So far we have proved that we can compress any finite expansion given in terms of blocks into a triangle, a block structure with depth $O(n)$ where $n$ is the height of the triangle. When translated into a quantum circuit, this compression corresponds to a compression of a long circuit into a short, fixed depth circuit. For NISQ devices, it is crucial to have short depth circuits due to the fact that the noise increases for longer circuit depths. In the rest of this section, we will show that the depth of the triangle can be reduced in half, leading to better quantum circuits for NISQ devices.

\begin{lemma}\label{lem:cascshift}
    For $i<a \leq b \leq j$,
    \begin{align}
        C_{a,b} \: C_{i,j} = C_{i,j} \: C_{a-1,b-1}.
    \end{align}
\end{lemma}
\begin{proof}
$C_{a,b} = B_a B_{a+1} \cdots B_b$ is formed by blocks with indices that are between $i$ and $j$, hence each of them can be passed from left side of $C_{i,j}$ to right side by decrementing the index by one.
\end{proof}

\begin{theorem}\label{thm:triangle_to_square}
    A triangle with height $n$ can be further optimized to a square with height $n$ which is defined as
    \begin{align}
        S_n = \prod_{\substack{k=n-1 \downarrow \\ k \mathrm{\ even}}}^2 C_{k,n} \prod_{\substack{k=n \downarrow \\ k \mathrm{\ odd}}}^1 C_{1,k}
    \end{align}
    for odd $n$ and
    \begin{align}
        S_n = \prod_{\substack{k=n-1 \downarrow \\ k \mathrm{\ odd}}}^1 C_{k,n} \prod_{\substack{k=n-1 \downarrow \\ k \mathrm{\ odd}}}^1 C_{1,k}
    \end{align}
    for even $n$, where down arrow $\downarrow$ indicates that the product is done in the decreasing order for index $k$. $S_n$ has the same number of blocks as $T_n$, however it has two-thirds the depth of $T_n$ which is important for quantum computers. 
\end{theorem}
\begin{proof}
First, consider the case for odd $n$. The definition of a triangle is:
\begin{align}
    T_n = \prod_{k=0}^{n-1} C_{n-k,n} = C_{n,n}\: C_{n-1,n} \cdots C_{1,n}
\end{align}
We will carry all $C_{\text{odd},n}$ cascades, i.e. $C_{n,n}$, $C_{n-2,n}$,\ldots $C_{5,n}$, $C_{3,n}$ to the right side of $C_{1,n}$ by using \cref{lem:cascshift} in the given order. In that order, $C_{m,n}$ will pass through ${m-1}$ cascades and become $C_{1,n-m+1}$ at the end, thereby leaving all the $C_{\text{even},n}$ cascades behind. This leads to a square formation for odd $n$ as defined in the theorem.

For even $n$, we will carry all $C_{\text{even},n}$ cascades, i.e. $C_{n,n}$, $C_{n-2,n}$,\ldots,$C_{4,n}$, $C_{2,n}$. In that order, $C_{m,n}$ will pass through ${m-1}$ cascades and become $C_{1,n-m+1}$ at the end, thereby leaving all the $C_{\text{odd},n}$ cascades behind. This leads to a square formation for even $n$ as defined in the theorem.
\end{proof}

In summary, triangles are useful to compress a Trotter expansion if the Trotter steps can be represented with blocks. After the compression, we are left with a triangle circuit. This can then be further simplified to a square (with same number of blocks as the triangle), but with half the total circuit depth. This leads to the following overall compression algorithm:

\begin{enumerate}
    \item Find a block mapping, and represent the circuit via blocks.
    \item Compress all blocks one by one into a triangle.
    \item Transform the triangle into a square.
\end{enumerate}
In the next section, we provide block mappings for Trotter decompositions of certain Hamiltonians.

\section{Time evolution of Certain Spin Models}
\label{sec:model_specific}

In this section, we show that the operations outlined in \cref{sec:compression_theorem}
are applicable to time evolution for several commonly found models in physics and
quantum information theory.
Each of these models has a Hamiltonian written as a sum of Pauli operators
for $n$ qubits (or $n$ spin-\sfrac{1}{2} particles)
\begin{align}
    \ham(t) = \sum_j H_j(t) \sigma^j,
\label{eq:hamiltonian}
\end{align}
where $H_j(t)$ are time- and site-dependent
coefficients and $\sigma^j$ are Pauli string operators, \textit{i.e.}, elements of the $n$-site Pauli group $\mathcal{P}_n=\{I,X,Y,Z\}^{\otimes n}$. 
The common feature of the models for which \cref{thm:blocks} applies is that they can all can be mapped onto free fermionic models after a Jordan-Wigner transformation \cite{chapman2020characterization} --- or, equivalently that the algebra generated by the operators in the Hamiltonian scales polynomially in the system
size \cite{kokcu2021fixed}. This suggests that there may be other models that share this property that will be amenable to this
approach.

We next consider the Trotter decomposition of the time evolution operator given in \cref{eq:solution} for those models. To generate a circuit for the full $U(t)$, we must convert each exponential $\exp(-i\Delta t \mathcal{H}(t_i))$ into a set of gates.  
The most straightforward way to map these exponentials into a set of gates is to first decompose the Hamiltonian in each exponential into components that each contain mutually commuting operators.  Trotterization can once again be employed to approximate each exponential as a product of exponentials of each component of the Hamiltonian. To the first order in $\Delta t$, this approximation is written as
\begin{align}\label{eq:approximation}
    e^{-i\Delta t \mathcal{H}(t_i)} = \prod_j e^{-i\Delta t H(t_i) \sigma^j} + O(\Delta t^2) 
\end{align}
where $j$ multiplies over the different components of the Hamiltonian.
Though higher order approximation schemes exist \cite{suzuki1976generalized}, this first order approximation is sufficient for our method. Once the time evolution is Trotter-decomposed into a series
of exponentials of single Pauli strings, a quantum circuit can then be constructed from single- and two-qubit operations \cite{pauli_circ1,pauli_circ2}. After finding a block mapping to exponentials of single Pauli strings that occur in \cref{eq:approximation} , the operations given in \cref{sec:compression_theorem} can be used to recombine the Trotter steps into a square, which can then be transformed into a fixed depth circuit for the models discussed below.

\begin{figure*}[htpb]
    \centering
    \includegraphics[width = 1.8\columnwidth]{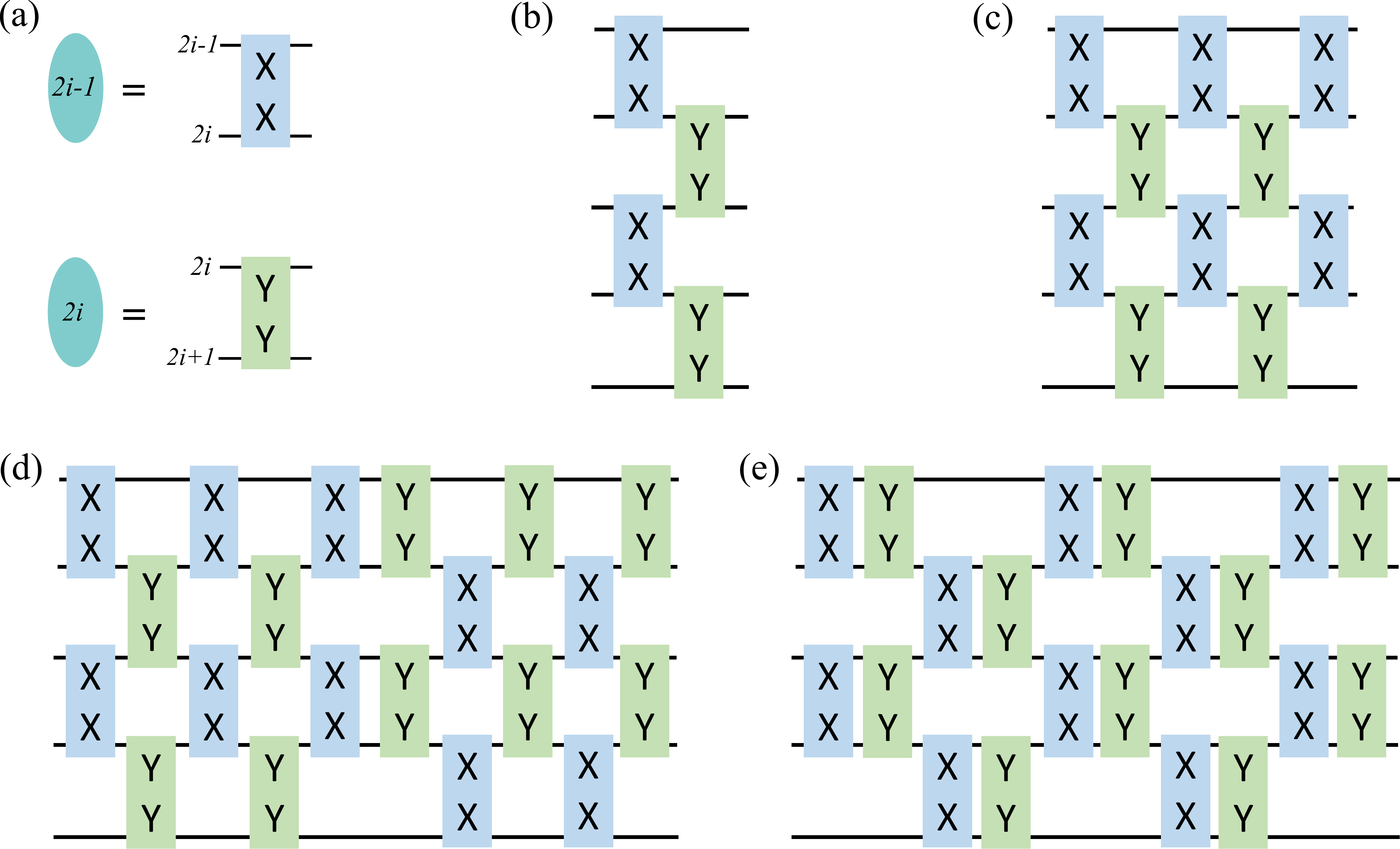}
    \caption{   \textbf{(a)} block mapping of 1D Kitaev chain, \textbf{(b)} the Trotter-decomposed operator for a single time step given in Eq.~\eqref{eq:kitaev_trotter} and \textbf{(c)} full time evolution circuit after Trotter compression. Each block $XX$ indicates $\exp(-i a XX)$ for some $a$ (and similar for $YY$). \textbf{(d)} Full time evolution of the XY model can be written as two consecutive evolutions of Kitaev chains. \textbf{(e)} These two evolutions can be combined to make $XX$ and $YY$ operations adjacent to each other to reduce the CNOT count by one half.   
    }
    \label{fig:kitaev}
\end{figure*}

\subsection{1D Kitaev Chain}\label{1D_Kitaev_Chain}

The Hamiltonian for the 1D Kitaev chain with open boundary conditions is given by
\begin{align}
    \ham = \sum_{\text{odd } i} a_i X_i X_{i+1} + \sum_{\text{even } i} b_i Y_i Y_{i+1}. 
\end{align}
The time evolution can be approximated using the Trotter decomposition with a time step $\Delta t$; for
a $5$-site chain this can be written explicitly as
\begin{align}
    e^{-i\Delta t \ham} =e^{-i\Delta t a\: X_1 X_2}e^{-i\Delta t c\: X_3 X_4}e^{-i\Delta t b\: Y_2 Y_3}e^{-i\Delta t d\: Y_4 Y_5}.
    \label{eq:kitaev_trotter}
\end{align}
Note that other choices are also available as multiplication of these exponentials in any order are all identical to each other within $O (\Delta t^2)$ error.

Consider the mapping 
$B_i \equiv B_i(\theta) = \exp(-i\theta X_i X_{i+1})$ for odd $i$ and $B_i \equiv B_i(\theta) = \exp(-i\theta Y_i Y_{i+1})$ for even $i$.
With this mapping, the Trotter step given in \cref{eq:kitaev_trotter} can be written in terms of $B_i$ matrices
\begin{align}
    e^{-i\Delta t \ham} =B_1(\Delta t \, a)\,B_3(\Delta t \, c)\,B_2(\Delta t \, b)\,B_4(\Delta t \, d).
\end{align}
Here we show that these operators are blocks and satisfy the properties given in \cref{block}.
Fusion is easily satisfied 
\begin{align}
\begin{split}
    e^{-ia X_i X_{i+1}}e^{-ib X_i X_{i+1}} &= e^{-i(a+b) X_i X_{i+1}}, \\
    e^{-ia Y_i Y_{i+1}}\:e^{-ib Y_i Y_{i+1}}\: &= e^{-i(a+b) Y_i Y_{i+1}},
\end{split}
\end{align}
and independent blocks that do not share a spin 
commute due to the fact that they act on different spins. This leaves us with the proof of the turnover property of blocks; we need to demonstrate that
\begin{align}\label{kitaevturnover}
\begin{split}
    e^{-ia X_i X_{i+1}}&e^{-ib Y_{i+1} Y_{i+2}}e^{-ic X_i X_{i+1}} \\
    &= e^{-i\alpha Y_{i+1} Y_{i+2}}e^{-i\beta X_i X_{i+1}}e^{-i\gamma Y_{i+1} Y_{i+2}}
\end{split}
\end{align}
where $i$ is an odd number, and $\alpha, \beta, \gamma$ are functions of $a,b,c$.
To show this,
we note that the exponentiated operators form an $\mathfrak{su}(2)$ algebra.
This is easily seen from their commutation relations,
\begin{align}\label{kitaevsu2}
\begin{split}
    [X_i X_{i+1},Y_{i+1} Y_{i+2}] &= 2i \: X_i Z_{i+1} Y_{i+2}, \\
    [X_i Z_{i+1} Y_{i+2},X_i X_{i+1}] &= 2i \: Y_{i+1} Y_{i+2}, \\
    [Y_{i+1} Y_{i+2},X_i Z_{i+1} Y_{i+2}] &= 2i \: X_i X_{i+1}. 
\end{split}
\end{align}
With this, we can leverage a standard Euler angle decomposition in $\mathfrak{su}(2)$.
Denoting the three directions in $\mathfrak{su}(2)$ as $\hat{x}$, $\hat{y}$, and $\hat{z}$,
any element in $\mathfrak{su}(2)$ can be written as a rotation around $\hat{x}$ followed by a rotation around $\hat{z}$ followed by another rotation around $\hat{x}$. 
In Eq.~\eqref{kitaevturnover}, the directions are $\hat{x}\equiv X_i X_{i+1}$, $\hat{z} \equiv Y_{i+1} Y_{i+2}$ on the left hand side, and vice versa on the right.
Therefore for any $a,b,c \in \mathbb{R}$, there exist $\alpha,\beta,\gamma \in \mathbb{R}$ and vice versa such that Eq.~\eqref{kitaevturnover} is satisfied. 

The coefficients can be calculated by mapping $2\times 2$ Pauli matrices to the exponents $X_i X_{i+1}$, $Y_{i+1} Y_{i+2}$ due to the isomorphism between the algebra generated by the exponents and $\mathfrak{su}(2)$. Mapping $X$ to $X_i X_{i+1}$ and $Z$ to $Y_{i+1} Y_{i+2}$, \cref{kitaevturnover} becomes
\begin{align}
\begin{split}
&\begin{pmatrix}
  \cos a & -i \sin a\\ 
  -i \sin a & \cos a
\end{pmatrix}
\begin{pmatrix}
  e^{-i b} & 0\\ 
  0 & e^{i b}
\end{pmatrix}
\begin{pmatrix}
  \cos c & -i \sin c\\ 
  -i \sin c & \cos c
\end{pmatrix} \\
&= 
\begin{pmatrix}
  e^{-i \alpha} & 0\\ 
  0 & e^{i \alpha}
\end{pmatrix}
\begin{pmatrix}
  \cos \beta & -i \sin \beta\\ 
  -i \sin \beta & \cos \beta
\end{pmatrix}
\begin{pmatrix}
  e^{-i \gamma} & 0\\ 
  0 & e^{i \gamma}
\end{pmatrix},
\end{split}
\end{align}
which may be solved to obtain the following equations

\begin{align}\label{eq:su2turnover_angles}
\begin{split}
\tan(\alpha + \gamma) & =
\phantom{-}\tan(b) \, \frac{\cos(a - c)}{\cos(a + c)}, \\
\tan(\alpha - \gamma) & =
-\tan(b) \, \frac{\sin(a - c)}{\sin(a + c)},
\end{split}
\end{align}
and their inverses
\begin{align}\label{eq:su2turnover_angles2}
\begin{split}
\tan(a + c) & =
\phantom{-}\tan(\beta) \, \frac{\cos(\alpha - \gamma)}{\cos(\alpha + \gamma)}, \\
\tan(a - c) & =
-\tan(\beta) \, \frac{\sin(\alpha - \gamma)}{\sin(\alpha + \gamma)}.
\end{split}
\end{align}
$\alpha$ and $\gamma$ can be solved directly from \cref{eq:su2turnover_angles}, and then $\beta$ can be solved from one of the relations given in \cref{eq:su2turnover_angles2}. Although we provide the analytical expressions here, in the code \cite{f3c,f3cpp} we use another method based on diagonalization of $2\times 2$ matrices, explained in detail in Ref.~\cite{simax}. That method is far more efficient and precise compared to calculating the angles via \cref{eq:su2turnover_angles,eq:su2turnover_angles2} due to the fact that the analytic expressions rely on repeated use of trigonometric and inverse trigonometric functions.

\begin{figure*}[htpb]
    \centering
    \includegraphics[width = 1.9\columnwidth]{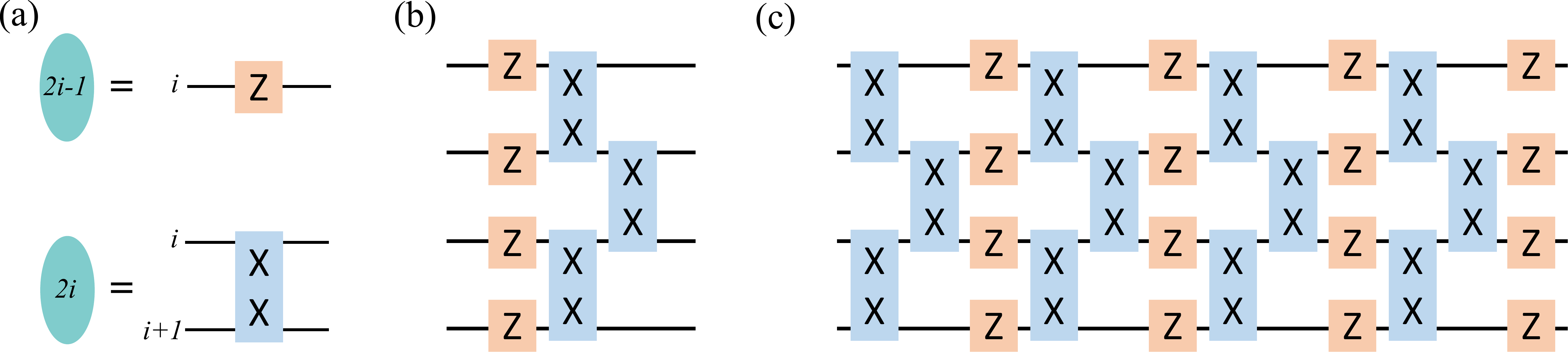}
    \caption{
    \textbf{(a)} block mapping of the 1D TFIM. As can be seen, blocks are not always 2 qubit operators.  \textbf{(b)} The Trotter-decomposed operator for a single time step of the evolution under the 1D TFIM Hamiltonian. \textbf{(c)} Full time evolution circuit after Trotter compression. Each block $XX$ indicates $\exp(-i a XX)$ for some $a$, and each block $Z$ indicates $\exp(-i b Z)$ for some $b$.}
    
    \label{fig:trotter_tfim}
\end{figure*}

Having shown that the elements of the Kitaev chain circuit satisfy the properties
of a block, \cref{thm:blocks,timedependentcompression,timeindependentcompression}  say that 
the Trotter circuit for this model can be compressed into a triangle of blocks; and \cref{thm:triangle_to_square} states that it can be further simplified into a square as given in \cref{fig:kitaev}(c).  
Considering the $n$ spin Kitaev chain, one Trotter step would have $n-1$ 2-qubit $XX$ and $YY$ rotations, leading to $n-1$ blocks with the mapping we provided. This would lead to $n(n-1)/2$ blocks in the final square circuit. Depending on the hardware, this means we either need $n(n-1)/2$ 2-qubit rotations, or considering the circuit implementation of $XX$ and $YY$ rotations via CNOT gates, we need $n(n-1)$ CNOT gates.

\subsection{XY Model}

The Hamiltonian for the 1D XY Model with open boundary conditions is given by 
\begin{align}
    \ham = \sum_{i=1}^{n-1} \big( a_i X_i X_{i+1} + b_i Y_i Y_{i+1} \big)
\end{align}
where $n$ is the number of spins. It can be shown that this Hamiltonian is the sum of two Kitaev chains
\begin{align}
    \ham = \ham_{\text{even X}} + \ham_{\text{odd X}},
\end{align}
where
\begin{align}
\begin{split}
    \ham_{\text{even X}} &= \sum_{\text{even } i} a_i X_i X_{i+1} + \sum_{\text{odd } i} b_i Y_i Y_{i+1}, \\
    \ham_{\text{odd X}} &= \sum_{\text{odd } i} a_i X_i X_{i+1} + \sum_{\text{even } i} b_i Y_i Y_{i+1}.
\end{split}
\end{align}
One can check that these chains are independent from each other, i.e. $[\ham_{\text{even X}},\ham_{\text{odd X}}] = 0$, by observing
\begin{align}\label{eq:kitaevscommute}
\begin{split}
    [X_i X_{i+1}, Y_{i}Y_{i+1}] &= 0, \\
    [X_i X_{i+1}, X_{i+1} X_{i+2}] &= 0, \\
    [Y_i Y_{i+1}, Y_{i+1} Y_{i+2}] &= 0. \\
\end{split}
\end{align}
Therefore the evolution of the XY model can be written as two separate evolutions of Kitaev chains, as given in \cref{fig:kitaev}(d). In fact, \cref{eq:kitaevscommute} also shows that any two terms in the separate evolution circuits commute with each other, and therefore can be moved freely. Considering the fact that $XX$ and $YY$ when individually implemented require $4$ CNOTs in total whereas their combination requires only $2$ CNOTs, it is advantageous to bring $XX$ and $YY$ rotations together for hardware that uses CNOTs. \cref{eq:kitaevscommute} allows us to transform the circuit in \cref{fig:kitaev}(d) to \cref{fig:kitaev}(e) and reduce number of CNOTs by one half.

\subsection{Transverse Field Ising Model}

The Hamiltonian for the 1D TFIM with open boundary conditions is given by
\begin{align}
    \ham = \sum_{i=1}^{n-1}a_i X_i X_{i+1} + \sum_{i=1}^n b_i Z_i.
\end{align}
The time evolution is again factored into Trotter steps of length $\Delta t$; for
a $4$-site chain this can be implemented as shown in \cref{fig:trotter_tfim}.

Consider the mapping $B_{2i-1}$ $\equiv$ $B_{2i-1}(\theta)$ $=$ $\exp(-i\theta Z_i )$ and $B_{2i}$ $\equiv$ $B_{2i}(\theta)$ $=$  $\exp(-i\theta X_i X_{i+1})$. With this mapping, a Trotter step for this model can be written as a product of $B_i$ matrices. Here we show that these operators satisfy the block properties from \cref{block}. Note that this is an example where not all blocks are 2-qubit operators, which shows that blocks can vary in size. The fusion operation is satisfied by this mapping
\begin{align}
\begin{split}
    e^{-ia X_i X_{i+1}}e^{-ib X_i X_{i+1}} &= e^{-i(a+b) X_i X_{i+1}}, \\
    e^{-ia Z_i}\:e^{-ib Z_i}\: &= e^{-i(a+b) Z_i}.
\end{split}
\end{align}
The commutation of blocks with an odd index $i$ is clear
since they do not act on a common spin: $B_{2i-1}$ acts on qubit $i$ and $B_{2i+1}$ acts on qubit $i+1$.
Blocks with even $i$ on the other hand should be checked explicitly, because $B_{2i}$ acts on spins $i$ and $i+1$, and $B_{2i+2}$ acts on spins $i+1$ and $i+2$, i.e., they both act on spin $i+1$. However, since $[X_i X_{i+1},X_{i+1} X_{i+2}]=0$, commutation of independent blocks applies for even indexed blocks as well.
This leaves us with the turnover property,
\begin{align}\label{tfimturnover}
\begin{split}
    e^{-ia X_i X_{i+1}}&e^{-ib Z_i}e^{-ic X_i X_{i+1}} \\
    &= e^{-i\alpha Z_{i}}e^{-i\beta X_i X_{i+1}}e^{-i\gamma Z_{i}}.
\end{split}
\end{align}
To show this, we note that similar to the Kitaev chain,
the exponents form an $\mathfrak{su}(2)$ algebra
\begin{align}\label{tfimsu2}
\begin{split}
    [Z_i,X_i X_{i+1}] &= 2i \: Y_i X_{i+1}, \\
    [Y_i X_{i+1},Z_i] &= 2i \: X_i X_{i+1}, \\
    [X_i X_{i+1},Y_i X_{i+1}] &= 2i \: Z_i . 
\end{split}
\end{align}
Therefore, by the arguments given in the discussion for the Kitaev chain, it can be shown that the turnover property is satisfied for the blocks as defined for the TFIM as well. The relation between the sets of angles are also the same as in \cref{eq:su2turnover_angles} due to the fact that it is the same switch between Euler decompositions of $\mathfrak{su}(2)$.

Having shown that the elements of the TFIM circuit satisfy the properties of a block (\cref{block}), it follows from  \cref{thm:blocks} that the Trotter circuit for this model can be compressed into a triangle of blocks.

Considering the $n$ spin TFIM, one Trotter step would have $n-1$ 2-qubit $XX$ and $n$ 1-qubit $Z$ rotations, leading to $2n-1$ blocks with the mapping we provided. This would lead to $n(2n-1)$ blocks or $n$ Trotter steps in the final square circuit. Depending on the hardware, this means we either need $n(n-1)$ $XX$ rotations, or considering the circuit implementation of $XX$ rotations via CNOT gates, we need $2n(n-1)$ CNOT gates.

\begin{figure*}[htpb]
    \centering
    \includegraphics[width = 2\columnwidth]{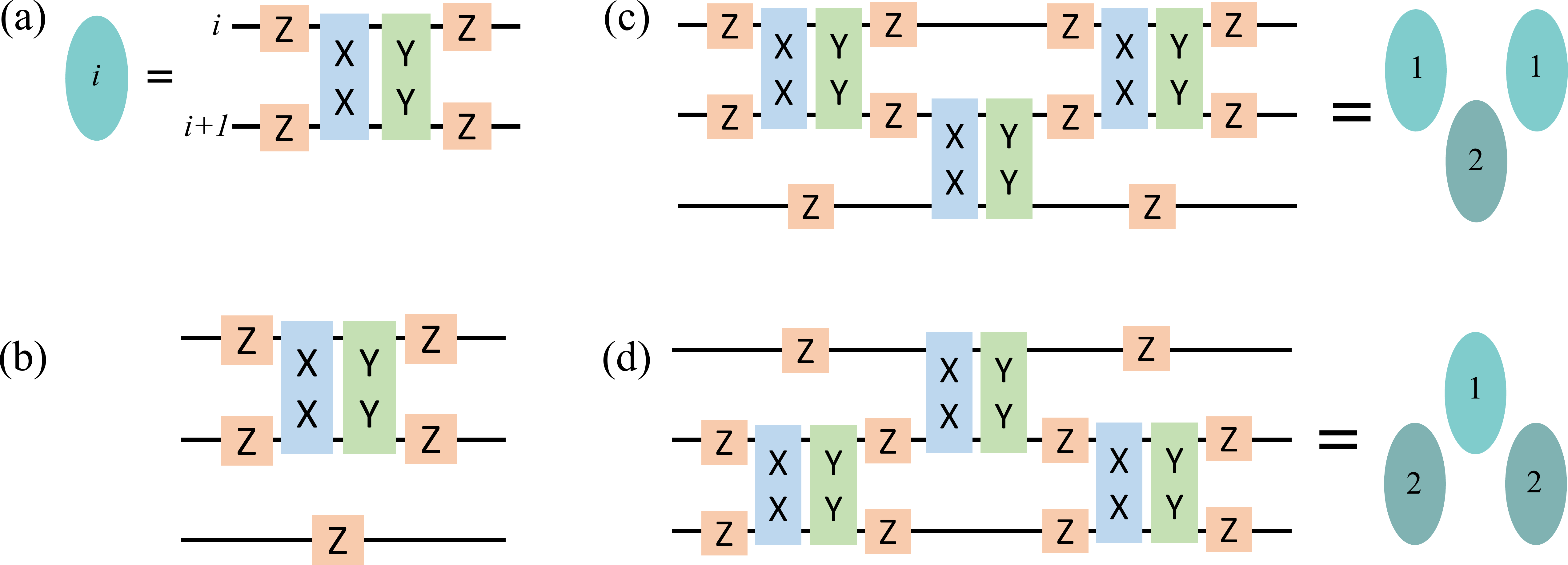}
    \caption{
    \textbf{(a)} block structure for the 1D TFXY model. \textbf{(b)} A generic circuit element that can represent any element in the group $\exp(\kalg)$ where $\kalg$ results from Cartan decomposition of $\galg(\text{TFXY}_3)$. It can be seen that this circuit represent a generic element via Eq.\eqref{eq:relation_k_tfxy2}. \textbf{(c)} Circuit structure for a generic element of the group $\exp(\galg(\text{TFXY}_3))$ obtained via Cartan decomposition corresponding to the involution $\theta(g) = Z_3 \: g \: Z_3$ and Cartan subalgebra given in Eq.~\eqref{eq:tfxy3_subalgebra}. As it can be seen, with the block definition, this structure corresponds to $B_1 B_2 B_1$.  \textbf{(d)} We could have chosen $\theta(g) = Z_1 \: g \: Z_1$ instead, and come up with a generic representation of an element in the group $\exp(\galg(\text{TFXY}_3))$ which corresponds to $B_2 B_1 B_2$. Considering that both panel (c) and panel (d) are generic representations of the same group, this proves that the structure given in panel (a) satisfies the turnover property.
    }
    \label{fig:tfxy}
\end{figure*}

\subsection{Transverse Field XY Model}

The Hamiltonian for the 1D transverse field XY (TFXY) model with open boundary conditions is given by

\begin{align}
    \ham = \sum_{i=1}^{n-1}(a_i X_i X_{i+1} + b_i Y_i Y_{i+1}) + \sum_{i=1}^n c_i Z_i.
\end{align}
Here, the verification of the block properties is somewhat more complex. First, we define
the $B_i$ operator as
\begin{align}\label{eq:tfxy_block_mapping}
\begin{split}
    B_{i} \equiv& e^{-i a\: Z_i}e^{-i b\: Z_{i+1}}e^{-i c\: X_i X_{i+1}}\\
    &e^{-i d\: Y_i Y_{i+1}}\:e^{-i f\: Z_i}e^{-i g\: Z_{i+1}}    
\end{split}
\end{align}
which can be shown diagrammatically as in \cref{fig:tfxy}(a). A Trotter step for this model can be written in terms of the $B_i$ matrices, since the exponents in \cref{eq:tfxy_block_mapping} involve all the terms in the Hamiltonian. We proceed to prove that the mapping \eqref{eq:tfxy_block_mapping} satisfies all block properties.

\emph{Commutation.} First, we prove that the $B_i$ operators satisfy the commutation relation given in \cref{block}. This property follows from the observation that if two blocks are not on top of each other nor nearest neighbors to each other, they do not act on a shared spin. For example, $B_{i}$ acts on qubits $i$ and $i+1$, and $B_{i+2}$ acts on qubits $i+2$ and $i+3$. Therefore they act on different spaces, and commute.

\emph{Fusion.} To prove the fusion property, it suffices to show that the expression given in Eq.~\eqref{eq:tfxy_block_mapping} represents a generic element of a Lie group.
The algebra formed by the exponents in Eq.~\eqref{eq:tfxy_block_mapping} and their commutators is the Hamiltonian algebra for the 2 site TFXY model \cite{kokcu2021fixed} 
(here
we suppress $i=1$ for convenience):
\begin{align}\label{eq:g(2tfxy)}
    \galg(\text{TFXY}_2) = \mathrm{span} \: i \{ X_1 X_2, Y_1 Y_2, X_1 Y_2, Y_1 X_2, Z_1, Z_2 \}.
\end{align}
We can use an involution $\theta: \galg \rightarrow \galg$ 
to form a Cartan decomposition  $\galg = \kalg \oplus \malg$ \cite{raczka1986theory,*gilmore2008lie,*hall2015lie}. Specifically,
the involution $\theta(g) = -X_1 X_2 \:g^T \:X_1 X_2$ (which counts the number of $Z$ matrices),
where $g \in \galg(\text{TFXY}_2)$,
yields
\begin{align}
\begin{split}
    \kalg &= \mathrm{span} \: i \{ Z_1, Z_2 \}, \\
    \malg &= \mathrm{span} \: i \{ X_1 X_2, Y_1 Y_2, X_1 Y_2, Y_1 X_2 \}.
\end{split}
\end{align}
We can choose a maximal Abelian subalgebra $\halg \in \malg$ as
\begin{align}
    \halg = \mathrm{span} \: i \{ X_1 X_2, Y_1 Y_2\}.
\end{align}
The ``KHK theorem'' \cite{raczka1986theory,*gilmore2008lie,*hall2015lie} states that any element in the
group $G \in \exp(g)$ for $g\in \galg(\text{TFXY}_2)$ can be written as
\begin{align}\label{eq:tfxy2grouT_elem}
\begin{split}
    G =&  e^{-i a\: Z_1}e^{-i b\: Z_2}e^{-i c\: X_1 X_2}\\ &e^{-i d\: Y_1 Y_2}\:e^{-i f\: Z_1}e^{-i g\: Z_2},    
\end{split}
\end{align}
which is precisely Eq.~\eqref{eq:tfxy_block_mapping}. Thus,
two blocks can be fused because the right hand side of 
Eq.~\eqref{eq:tfxy_block_mapping} is a generic element in the Lie group 
$\exp(\galg(\text{TFXY}_2))$.

\emph{Turnover.} To show that the blocks $B_i$ satisfy the turnover property, 
we consider the Hamiltonian algebra of the 3 spin TFXY model,
\begin{align}\label{eq:g(3tfxy)}
\begin{split}
    \galg(\text{TFXY}_3) = \mathrm{span} \: i \{ &X_1 X_2, Y_1 Y_2, X_1 Y_2, Y_1 X_2, X_2 X_3,\\
    &Y_2 Y_3, X_2 Y_3, Y_2 X_3 , Z_1, Z_2 , Z_3,\\
    &X_1 Z_2 X_3, X_1 Z_2 Y_3, Y_1 Z_2 X_3, \\
    &Y_1 Z_2 Y_3  \}.
\end{split}
\end{align}
Using the involution $\theta(g) = Z_3 g Z_3$, 
where $g \in \galg(\text{TFXY}_3)$,
we partition
$\galg(\text{TFXY}_3)$ into
\begin{align}
\begin{split}
    \kalg = \mathrm{span} \: i \{&X_1 X_2, Y_1 Y_2, X_1 Y_2, Y_1 X_2, Z_1, Z_2, Z_3 \},\\
    \malg = \mathrm{span} \: i \{ &X_2 X_3, Y_2 Y_3, X_2 Y_3, Y_2 X_3, X_1 Z_2 X_3,\\
    &X_1 Z_2 Y_3, Y_1 Z_2 X_3, Y_1 Z_2 Y_3  \}.
\end{split}
\end{align}
We note that $\kalg$ is simply
\begin{align}\label{eq:relation_k_tfxy2}
    \kalg= (I_3 \otimes \galg(\text{TFXY}_2)) \oplus \mathrm{span} \: i \{ Z_3\}, 
\end{align}
where $I_3$ is the identity operator for the 3rd spin. The relation between $\kalg$ and $\galg(\text{TFXY}_2)$ together with Eq.~\eqref{eq:tfxy2grouT_elem} 
allows us to represent any element in the group  $\exp(\kalg)$ as in \cref{fig:tfxy}(b). Therefore by choosing the following Cartan subalgebra
\begin{align}\label{eq:tfxy3_subalgebra}
    \halg = \mathrm{span} \: i \{ X_1 X_2 , Y_1 Y_2 \},
\end{align}
the KHK theorem implies that any element $G$ in the group $\exp(\galg(\text{TFXY}_3))$ can be
written as $G = \exp(k_1) \exp(h) \exp(k_2)$ where $k_1,k_2 \in \kalg$ and $h \in \halg$. This may be
represented by the circuit given in  \cref{fig:tfxy}(c), which is a V shaped 3 block structure, with each
block representing the operator in Eq.~\eqref{eq:tfxy_block_mapping}.

The Cartan decomposition is not unique --- we could have alternatively used
the involution
$\theta(g) = Z_1 \: g \: Z_1$. 
This results in a swap of the 1st and 3rd spins. Thus, any element in the group $\exp(\galg(\text{TFXY}_3))$ can be equally well 
represented by an inverted circuit as shown in \cref{fig:tfxy}(d). 
This means that we have found two equivalent block structures that both represent a generic element
in the group $\exp(\galg(\text{TFXY}_3))$; these can then be transformed into
one another, which proves the turnover property for the block mapping of
Eq.~\eqref{eq:tfxy_block_mapping}.

The above proofs for both fusion and turnover properties for the block structure \cref{eq:tfxy_block_mapping} are based on the algebraic properties of $\galg(\text{TFXY}_2)$ and $\galg(\text{TFXY}_3)$, and do not show how to calculate the coefficients after the fusion or turnover. Here we provide another Lie algebra based method to calculate the coefficients after the fusion and turnover operations.

\emph{Fusion calculation.} We already proved that the $B_i$ matrices given in \cref{eq:tfxy_block_mapping} form a Lie group. The Lie algebra associated with that group turned out to be the Hamiltonian algebra of a 2 site TFXY model \eqref{eq:g(2tfxy)}. We can use a different basis for that algebra, and rewrite it as
\begin{align}
\begin{split}
    \galg(\text{TFXY}_2) = \mathrm{span} \: i \{& (X_1 X_2 \pm Y_1 Y_2)/2, \\ 
    &(X_1 Y_2 \pm Y_1 X_2)/2,\\ 
    &(Z_1 \pm Z_2)/2 \}.
\end{split}
\end{align}
We can split it into two algebras $\galg_+$ and $\galg_-$ given by
\begin{align}
\begin{split}
    \galg_- = \mathrm{span} \: i \{& (X_1 X_2 - Y_1 Y_2)/2, (Z_1 + Z_2)/2, \\ 
    &(X_1 Y_2 + Y_1 X_2)/2 \},
\end{split}
\end{align}
and
\begin{align}
\begin{split}
    \galg_+ = \mathrm{span} \: i \{& (X_1 X_2 + Y_1 Y_2)/2, (Z_1 - Z_2)/2, \\ 
    &(X_1 Y_2 - Y_1 X_2)/2\}.
\end{split}
\end{align}
A simple check on the commutation relations for the basis elements yields that both algebras are isomorphic to $\mathfrak{su}(2)$, and both are independent, i.e. $[\galg_+,\galg_-]=0$. Renaming $(X_i X_{i+1}\pm Y_i Y_{i+1})/2\equiv Q^\pm_i$ and $(Z_i\pm Z_{i+1})/2\equiv R^\pm_i$ for convenience, one can rewrite \cref{eq:tfxy_block_mapping} as a multiplication of two independent $SU(2)$ elements  
\begin{align}
\begin{split}
    B_{i} \equiv& e^{-i (a+b)\: R^+_i}e^{-i (a-b)\: R^-_i}e^{-i (c+d)\: Q^+_i}\\
    &e^{-i (c-d)\: Q^-_i}e^{-i (f+g)\: R^+_i}e^{-i (f-g)\: R^-_i}  \\
    =& \big( e^{-i (a+b)\: R^+_i}e^{-i (c-d)\: Q^-_i}e^{-i (f+g)\: R^+_i} \big)\\
    &\big(e^{-i (a-b)\: R^-_i}e^{-i (c+d)\: Q^+_i}e^{-i (f-g)\: R^-_i} \big).
\end{split}
\end{align}
Then, the fusion operation will boil down to multiplication of two set of $SU(2)$ matrices represented by an Euler decomposition, which only requires the $\mathfrak{su}(2)$ turnover that we provided in \cref{eq:su2turnover_angles}.Each $SU(2)$ multiplication will require one $\mathfrak{su}(2)$ turnover operation, therefore fusion of two TFXY blocks require two $\mathfrak{su}(2)$ turnover operations in total.

\emph{Turnover calculation.} This proof entirely relies on the fact that the TFXY block \cref{eq:tfxy_block_mapping} can be written as a combination of TFIM blocks, i.e. $Z$ and $XX$ rotations. This can be achieved by plugging
\begin{align}
    Y_i Y_{i+1} = e^{i \frac{\pi}{4} Z_i} e^{i \frac{\pi}{4} Z_{i+1}} X_i X_{i+1} e^{-i \frac{\pi}{4} Z_i} e^{-i \frac{\pi}{4} Z_{i+1}},
\end{align}
into \cref{eq:tfxy_block_mapping}. Therefore a V shaped TFXY block structure $B_i\:B_{i+1}\:B_i$ can be transformed to a TFIM triangle with height 3 that acts on qubits $i$, $i+1$ and $i+2$ given as 
\begin{figure}[htpb]
    \centering
    \includegraphics[width = 0.9\columnwidth]{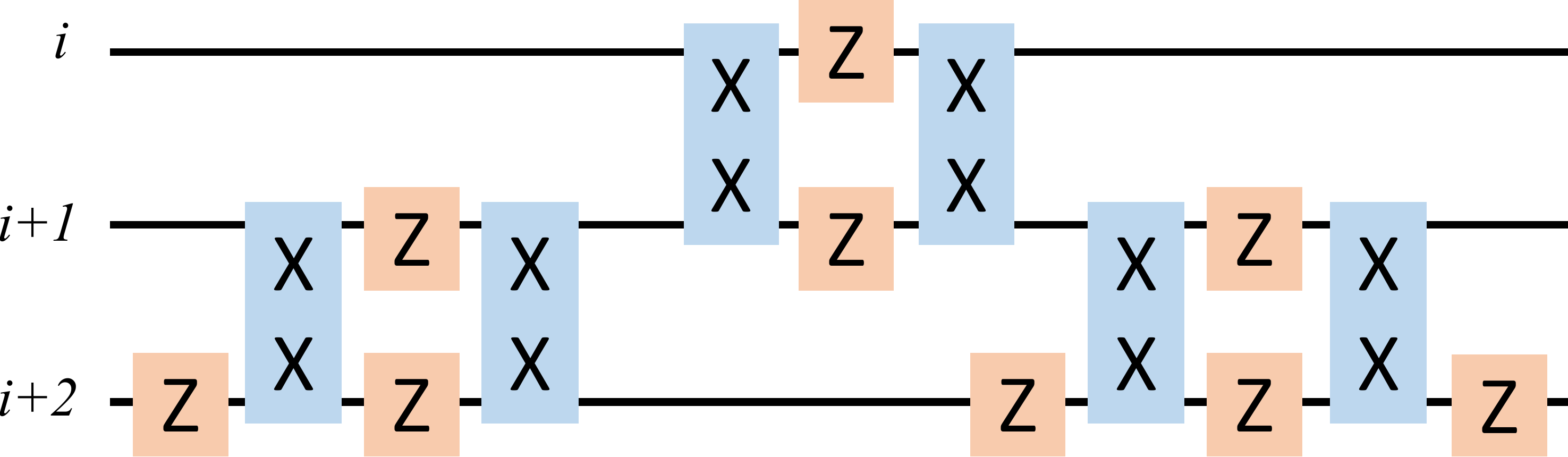}
\end{figure}
\\
\noindent where we deliberately grouped some gates together. All three groups can be rewritten as a TFXY block by applying the TFXY fusion operation we explained above once. For example, the leftmost group is a product of the block $B_{i+1} \equiv \exp(ia\:Z_{i+2}) \: \exp(ib\:X_{i+1} X_{i+2})$, which is a TFXY block with some parameters equal to $0$, with another block $B_{i+1} \equiv \exp(ic\:Z_{i+1}) \: \exp(id\:Z_{i+2}) \: \exp(ie\:X_{i+1} X_{i+2})$. Therefore, we end up with a $\Lambda$ shaped TFXY block structure. The reverse can be done in the same way by compressing $\Lambda$ into an upside down triangle instead.

To form the TFIM triangle from the V shaped TFXY block structure, 26 $\mathfrak{su}(2)$ turnover operations must be applied. To combine the 2-qubit pieces into a TFXY block, 3 TFXY merge operations must be applied. Each TFXY merge requires 2 $\mathfrak{su}(2)$ turnovers, leading to a total of 32 $\mathfrak{su}(2)$ turnover operations for one TFXY turnover operation.   

The constructive proofs we provide here use trigonometric and inverse trigonometric function evaluations frequently, which are more complicated than doing linear algebra operations. Due to this fact, we provide a linear algebra based method to evaluate the angles for the fusion and turnover operations of a TFXY model in \cite{simax}. Also with this method, the TFXY turnover can be performed at a cost equivalent to 20 $\mathfrak{su}(2)$ turnovers, which is another point on the efficiency. However, we include the Lie algebra based proof here as well because it provides additional insight about the algebraic structure of the TFXY model.

Since the mapping of the 1D TFXY model in \cref{eq:tfxy_block_mapping} satisfies the  
properties listed in \cref{block}
it can be compressed into a triangle according to \cref{trianglezigzag}.

For an $n$ spin TFXY model, one Trotter step has $n-1$ $XX$ and $YY$ rotations with $n$ 1-qubit $Z$ rotations, leading to $n-1$ blocks with the mapping we provided for the TFXY model. This leads to $n(n-1)/2$ blocks in the final square circuit. Depending on the hardware, this means we either need $n(n-1)$ 2-qubit rotations (2 per block), or considering the circuit implementation of adjacent $XX$ and $YY$ rotations via 2 CNOT gates \cite{vidal2004universal}, we need $n(n-1)$ CNOT gates.

\subsection{Generalized TFXY and Free Fermions}

The algebra generated by the elements in Eq.~\eqref{eq:g(2tfxy)} 
and the generic group representation in Eq.~\eqref{eq:tfxy2grouT_elem} imply that we can generalize the results for the TFXY model to the following Hamiltonian (which can be called generalized TFXY Hamiltonian)
\begin{align}
\begin{split}
    \ham =& \sum_{i=1}^{n-1}(a_i X_i X_{i+1} + b_i Y_i Y_{i+1}  \\ 
    & + c_i X_i Y_{i+1}+ d_i Y_i X_{i+1})+\sum_{i=1}^n f_i Z_i.
\end{split}
\end{align}
with the same block mapping from Eq.~\eqref{eq:tfxy_block_mapping}. Via the Jordan--Wigner transformation, this Hamiltonian can be mapped to
\begin{align}
\begin{split}
    \ham =& \sum_{i=1}^{n-1}(\alpha_i \hat{c}^\dagger_i \hat{c}_{i+1} + \alpha^*_i \hat{c}^\dagger_{i+1} \hat{c}_{i} \\ 
    &+ \beta_i \hat{c}_i \hat{c}_{i+1} + \beta^*_i \hat{c}^\dagger_{i+1} \hat{c}^\dagger_{i})+\sum_{i=1}^n \gamma_i \hat{c}^\dagger_i \hat{c}_{i}.
\end{split}
\end{align}
where $\hat{c}_{i} (\hat{c}^\dagger_i)$ is the fermion annihilation (creation) operator on site $i$. Therefore, using the block mapping in Eq.~\eqref{eq:tfxy_block_mapping}, we can compress Trotter decompositions under any free, one-dimensional Hamiltonian with nearest neighbor hopping and open boundary conditions.

\subsection{Comparison of the Models}

We introduced three different block mappings with two types of fusion and turnover operations. The turnover operation for the blocks given in the Kitaev chain, XY model and TFIM use Euler decompositions of $\mathfrak{su}(2)$, whereas the ones for TFXY and generalied TFXY models use algebraical properties of $\galg(\text{TFXY}_3)$. The $\mathfrak{su}(2)$ turnover is approximately 20 times more efficient than the $\galg(\text{TFXY}_3)$ turnover, because the former requires diagonalization of a $2\times2$ matrix whereas the latter requires simultaneous diagonalization of four $2\times2$ matrices \cite{simax}, although in practice this only matters when the system size grows large.

Although some models share the same type of fusion and turnover properties, 
the resulting circuits vary in the number of CNOT or 2-qubit XX/YY rotation gates required for $n$ qubits, as shown in \cref{tab:cnotcounts}
\begin{table}[h]
\begin{center}
\begin{tabular}{ c | c | c | c }
 Block Mapping& $XX+YY$  & \#CNOT & Turnover Type \\ 
 \hline
 Kitaev Chain & $\frac{n(n-1)}{2}$ & $n(n-1)$ & $\mathfrak{su}(2)$ \\  
 XY Model & $n(n-1)$ & $n(n-1)$ & $\mathfrak{su}(2)$\\
 TFIM & $n(n-1)$ & $2n(n-1)$& $\mathfrak{su}(2)$ \\
 TFXY & $n(n-1)$&$n(n-1)$ & $\galg(\text{TFXY}_3)$
\end{tabular}
\caption{Number of 2-qubit gates required for the constant depth
circuits on $n$ qubits for the models discussed. Here, $XX+YY$
indicates the total number of $XX$ and $YY$ rotations.}
\label{tab:cnotcounts}
\end{center}
\end{table}

Because all of these models can be derived from TFXY model, the circuits for the Kitaev chain, XY model, and TFIM do not have to be compressed by the model-specific block mappings we discussed above; one can consider them as specific instances of the TFXY model and use the TFXY block mapping instead. For the TFIM, if one uses the TFXY block structure, the height of the Trotter step zigzag (\cref{zigzag}) reduces by half, leading to 4 or 8 times fewer turnover operations depending on whether the system is time dependent or not due to  \cref{timedependentcompression,timeindependentcompression}. Although the TFXY turnover is 20 times more expensive than the $\mathfrak{su}(2)$ turnover, and this leads to $3-5$ times more time spent for the compression, using TFXY block structure leads to half the number of CNOT gates in the resulting circuits, which proves particularly useful for quantum computers that rely on the CNOT 2-qubit gate (such as transmon-based hardware).

\section{Example: Adiabatic State Preparation}
\label{sec:specific_example}

\begin{figure*}[htpb]
    \centering
    \includegraphics[width = 2.0\columnwidth]{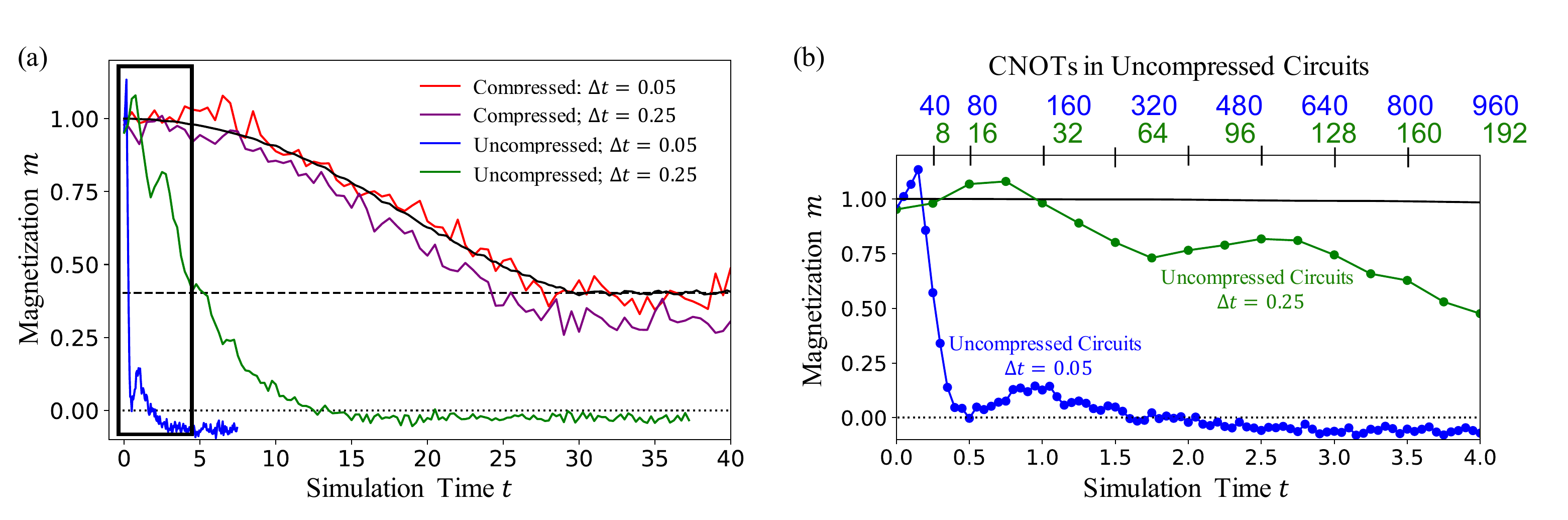}
    \caption{Simulation results from \emph{ibmq\_brooklyn} comparing compressed circuits versus uncompressed circuits derived from standard Trotter decomposition.  A 5-qubit system is evolved under a time-dependent Hamiltonian that is adiabatically varied from an Ising model to a transverse field Ising model (TFIM) from time $t=0$ to time $t=30$, and then evolved under the final (time-independent) TFIM for the remaining time. 
    (a) Simulation results from compressed circuits (red, purple) and uncompressed circuits (blue, green) for two different values of Trotter step-size compared to the numerically exact results (solid black line).  
    The black dashed line gives the expected magnetization of the final TFIM, which is maintained through the final evolution ($t > 30)$. The red curve shows higher fidelity than the purple curve due to smaller Trotter error. Results from the uncompressed circuits quickly decay to random noise, denoted by the dotted black line at a magnetization of 0.  
    (b) A zoomed-in view of the uncompressed circuit results within the portion of the plot in panel (a) in the solid black box.  CNOT counts for each circuit throughout the simulation are denoted on the top horizontal axis with color corresponding to the Trotter step-size used in the uncompressed circuits. The compressed circuits contain 20 CNOTs for all time-steps.
    }
    \label{fig:results}
\end{figure*}

An important step of many static and dynamic simulations problems on quantum computers involves preparing the ground state of the relevant Hamiltonian.  In general, this ground state is non-trivial to prepare on the qubits, even in cases where the state may be known analytically.  In this example, we demonstrate how our compressed circuits can be used to efficiently prepare non-trivial ground states on the quantum computer with high-fidelity via adiabatic state preparation (ASP) \cite{aspuru2005simulated, barends2016digitized}.  In ASP, qubits are prepared in the ground state of some initial Hamiltonian $\ham_I$, which is presumed to be trivial to prepare.  The system is then evolved under a parameter-dependent Hamiltonian $\ham(s) = (1-s)\ham_I + s\ham_P$ as $s$ is varied from 0 to 1. Here, $\ham_P$ is the problem Hamiltonian, whose ground state is desired and in general is non-trivial to prepare.  If $s$ is varied slowly enough, the system will remain in the ground state of the instantaneous Hamiltonian according to the adiabatic theorem \cite{born1928beweis}.  Just how slow this variation needs to be is dependent upon the size of the energy gap between the ground state and first excited state.  

Evolving the system under $\ham(s)$ as $s$ is varied from 0 to 1 can be viewed as evolving the system under a time-dependent Hamiltonian.  Therefore, assuming that the instantaneous Hamiltonian remains within the class of free-fermionic models while varying $s$, we can apply our compression techniques to efficiently produce minimal, fixed-depth circuits for such an adiabatic evolution.  Our technique offers an enormous advantage in ASP, as the ability to generate fixed-depth circuits in an efficient amount of time for an arbitrarily large number of time-steps allows us to vary $s$ arbitrarily slowly, thereby guaranteeing adiabatic evolution with one simple fixed-depth circuit. As shown in the companion paper \cite{simax}, the compression can be
performed efficiently, and thus poses no limitation
on the adiabaticity.

Without our compression techniques, the quantum circuits generated via standard Trotter decomposition rapidly become too large to produce high-fidelity results on NISQ devices after a certain number of time-steps \cite{smith2019simulating}.  With such a noise-imposed constraint on the total number of time-steps $N$, a compromise must be made between the total adiabatic evolution time $T$ and the size of the time-step $\Delta t$, as $N = T/\Delta t$. To minimize $N$, one can minimize $T$, although if the total time is made too short, the Hamiltonian is varied too quickly and adiabaticity will not be achieved.  Alternatively, one can maximize $\Delta t$, although if the time-step is made too large Trotter error will become too large.  By eschewing the constraint on $N$ altogether, our compressed circuits allow $T$ to be made as large as necessary to achieve adiabaticity and $\Delta t$ to be made as small as necessary to practically eliminate Trotter error, thereby enabling high-fidelity ASP within this class of Hamiltonians. 

As an example, we demonstrate ASP of a spatially uniform, 5-spin TFIM with open boundary conditions, comparing performance between compressed circuits and circuits generated with standard Trotter decomposition (i.e., uncompressed circuits).  The parameter-dependent Hamiltonian is given by $\ham(s) = (1-s)\ham_I + s\ham_P$ where $\ham_I = h \sum_{i=1}^{n} Z_i$ and $\ham_P = J_P \sum_{i=1}^{n-1} X_i X_{i+1} + h \sum_{i=1}^{n} Z_i$, where $n=5$ is the number of qubits.  

Varying the parameter $s$ from $0$ to $1$ is equivalent to the following time-dependent Hamiltonian:
\begin{align}
    \ham_{ASP}(t) = J(t) \sum_{i=1}^{n-1} X_i X_{i+1} + h \sum_{i=1}^{n} Z_i
    \label{eq:H_asp}
\end{align}
where $J(0)=0$, and $J(t)$ is increased linearly to a final time $t=T$ such that $J(T) = J_P$ as defined in the problem Hamiltonian.  

The qubits are initialized in the ground state of $\ham_I = \ham_{ASP}(0)$, which in this case is simply all qubits in the spin-up orientation, a trivial state to prepare on the quantum computer.  Next, we evolve the qubits under the time-dependent Hamiltonian $\ham_{ASP}$ from $t=0$ to $t=T$ using Trotter decomposition to break this evolution into small time-steps of size $\Delta t$. If $T$ is large enough to achieve an adiabatic evolution, and if $\Delta t$ is small enough to minimize Trotter error, the final state of the qubits after this evolution will be the ground state of our problem Hamiltonian $\ham_P = \ham_{ASP}(T)$, which in general is non-trivial to prepare.  Finally, we end the simulation by evolving the system under the time-independent problem Hamiltonian $\ham_P = \ham_{ASP}(T)$.  If adiabatic simulation was successful, the system will remain in the stationary ground state of the problem Hamiltonian for the entirety of this final evolution.  

We use \emph{ibmq\_brooklyn}, a quantum processor that relies on CNOT as a 2-qubit gate, to generate our results. To compress the Trotter circuit, we applied TFXY block structure rather than TFIM block structure to minimize the number of CNOTs as given in \cref{tab:cnotcounts}.
\Cref{fig:results} shows simulation results 
with $h = -1$, $J_P = -2$, and $T=30$ using compressed circuits (red and purple) and uncompressed circuits (green and blue). Both plots show the average magnetization along the $z$-direction of the 5-qubit system versus time as it evolved under $\ham_{ASP}$.   The average magnetization for an $n$-spin system is given by $\langle m(t)\rangle \equiv  \frac{1}{n} \sum_i \sigma_i^z(t)$.  The solid black line gives the exact system magnetization of the instantaneous Hamiltonian $\ham_{ASP}(t)$ (computed numerically via exact diagonalization), and the black dashed line shows the expected magnetization of the final Hamiltonian $\ham_P = \ham_{ASP}(T)$.  If the system is truly varied adiabatically, its magnetization should reach this value at $T=30$ and stay at this value throughout the remainder of the evolution under the problem Hamiltonian.  

The red and purple curves in \cref{fig:results}(a) show simulation results from compressed circuits executed on \emph{ibmq\_brooklyn} using two different values for $\Delta t$ used for the Trotter decomposition. We applied readout error mitigation and zero-noise extrapolation to these results.  The red curve is in remarkably good agreement with the exactly computed values, while the purple curve stabilizes at a final magnetization that is slightly smaller than expected;  this discrepancy is due to Trotter error.  As is clear from \cref{fig:results}, $\Delta t = 0.05$ is sufficiently small, while $\Delta t = 0.25$ leads to non-negligible Trotter error.  While $\Delta t = 0.05$ appears sufficiently small in terms of Trotter error, we note that our compressed circuits allow for $\Delta t$ to be arbitrarily small.

In comparison, the blue curve in \cref{fig:results}(a) shows the magnetization from the same quantum processor using uncompressed circuits and a time-step of $\Delta t = 0.05$.  We see an immediate loss of fidelity in the results as the uncompressed circuits quickly become too deep for NISQ hardware.  The magnetization quickly decays towards $0$, which indicates the results are simply random noise, denoted by the dotted black line.  To improve upon these results, we increased $\Delta t$, such that fewer time-steps are required to reach the same total simulation time, and which in turn means the circuit depths of the uncompressed circuits grows more slowly with the simulation time [\cref{fig:results}(b)].  This is shown as a green curve, which shows hardware results from uncompressed circuits with a time-step of $\Delta t = 0.25$. While its performance appears slightly better, these results, too, eventually decay to a magnetization of zero well before the adiabatic evolution is complete.  

\Cref{fig:results}(b) shows a zoomed-in view of the portion of the plot in panel a surrounded by the solid black box.  Quantum hardware simulation results for the uncompressed circuits for both $\Delta t$ values are shown, with the number of CNOTs in the circuit used for each time-step given on the top horizontal axis in the corresponding color.  The number of CNOTs per circuit for $\Delta t = 0.05$ grows much faster with total simulation time than for $\Delta t = 0.25$, which is why the green curve maintains higher-fidelity at initial simulation times.  For reference, our compressed circuits only contain 20 CNOTs for all time-steps for this 5-qubit system.

\section{Discussion and outlook}
\label{sec:discussion}
Our results have several implications.
First, we have shown that time evolution for the models above can be
performed with a fixed depth circuit, even for a time-dependent Hamiltonian. Note that if the entire circuit for time evolution does not
follow the required block structure, sub-circuits can still be optimized in this way (if the sub-circuits do satisfy the block structure).
While we have focused on the generalized TFXY model and its
descendants, the approach here should extend to any model that
has similarly structured frustration graphs \cite{chapman2020characterization},
or can be mapped onto a free fermionic Hamiltonian with nearest neighbor hopping.
We should note that the compression method cannot be universally applied,
as is evident from the No-Fast-Forwarding theorem \cite{atia2017fast},
but it is possible
that there are other classes of models where the compression method
works, such as those outlined by Gu et al. \cite{gu2021fast}.

Considering the fact that our method does not require the Hamiltonian to be translationaly invariant or time independent, the class of free fermionic Hamiltonians 
that can be compressed is relatively broad. For example,  free fermionic systems after a quench from a clean to a disordered system 
which show complex
dynamics of the entanglement entropy \cite{Lundgren2019}. Certain topological models, e.g. the staggered free fermionic system of the Rice-Mele model \cite{ricemele}, are also accessible. There, the adiabatic evolution is characterized by the Chern number of the system; although
NISQ quantum computers can measure the Chern number in certain cases \cite{xiao2021robust}, adiabatic evolution
requires a compressed circuit. Free fermionic systems also emerge as mean field approximations of many physical systems of interest such as superconductivity and charge density waves \cite{cdw, altland2010condensed}. With our method, a fixed depth circuit can be generated to simulate \emph{non-equilibrium} mean field theories.

Free fermionic time-evolution can be used for ground state preparation via adiabatic evolution, which we have chosen as an example application in this work. More generally,
it can be shown that any generic Bogolyubov transformation  \cite{Valatin1958ja,Bogolyubov1958km}, i.e. unitary transformations of one particle creation-annihilation operators, can be written as a free fermionic time-evolution. Therefore the compressed circuits can be used to generate fermionic Gaussian States as a Variational Quantum Eigensolver (VQE) ansatz for the ground state of any free fermionic Hamiltonian, or as a starting approximation for interacting systems. This was recently demonstrated for Slater determinants\cite{babush1}, Hartree-Fock anzatse\cite{babush2}, and generic fermionic Gaussian states\cite{babush3}. 
However, these take a global approach to the problem --- mapping the fermionic
evolution on $n$ sites to an $n \times n$ matrix, and producing the circuit
via Givens rotations. The compression method presented here relies only on \emph{local}
operations of the circuit elements; this obviates the need to work with large
matrices, and extending the range of applicability (as we do not rely on
a global matrix mapping). 

From a broader perspective, we have revealed a new method for
manipulating and compressing circuit elements that have three
relatively simple to check properties. For the time evolution problem
discussed in this work, the blocks
are composed of 1- and 2-qubit sets of operations that are derived from particular models--- but this is
not a requirement, and the circuit operations work equally well
for larger unitaries as long as the properties of blocks
are satisfied. The blocks also do not have to be identical (as can be seen in the Kitaev chain example) and they do not even have to be operating on the same number of qubits (as can
be seen in the TFIM example above). This suggests that the ideas outlined
above could be incorporated into transpiler software, and that the method we proposed is far more general than Givens rotation based methods listed above.

We anticipate that the ideas developed here can be
extended and applied beyond time evolution. In particular,
we are investigating whether circuits that arise from the Quantum Approximate Optimization Algorithm may be treated in this fashion. Similarly, an extension to higher dimensions may be possible, enabling simulations on quantum hardware far beyond
the current limitations.

\vspace{0.1in}
\begin{acknowledgments}
We acknowledge helpful conversations with Eugene Dumitrescu and Thomas Steckmann.
EK, JKF and AFK were supported by the Department of Energy, Office of Basic Energy Sciences, Division of Materials Sciences and Engineering under Grant No. DE-SC0019469. JKF was also supported by the McDevitt bequest at Georgetown University.
DC and RVB are supported by the Laboratory Directed Research and Development Program of Lawrence Berkeley National Laboratory under U.S. Department of Energy Contract No. DE-AC02-05CH11231.
LB and WAdJ were supported by the U.S. Department of Energy (DOE) under Contract No.  DE-AC02-05CH11231,  through the Office of Advanced Scientific Computing  Research  Accelerated  Research  for  Quantum Computing  Program.This research used
resources of the Oak Ridge Leadership Computing Facility, which is a
DOE Office of Science User Facility supported under Contract
No.~DE-AC05-00OR22725.
We acknowledge the use of IBM Quantum services for this work. The views expressed are those of the authors, and do not reflect the official policy or position of IBM or the IBM Quantum team.
\end{acknowledgments}

\bibliography{refs}
\bibliographystyle{apsrev4-2}

\end{document}